\documentclass[conference,a4paper]{IEEEtran}
\IEEEoverridecommandlockouts
\usepackage{epsfig, amsmath,amssymb,epsf,cite,subfigure,scalefnt,multicol}
\usepackage{graphicx}
\usepackage{color}
\usepackage{algorithm}
\usepackage{algpseudocode}
\usepackage{subfigure}
\usepackage{epstopdf}
  \def\cC{{\mathcal{C}}} 
\def\cE{{\mathcal{E}}} \def\cF{{\mathcal{F}}}  \def\cH{{\mathcal{H}}}
\def\cI{{\mathcal{I}}}   
\def\cM{{\mathcal{M}}} \def\cN{{\mathcal{N}}}  \def\cP{{\mathcal{P}}}
 \def\cR{{\mathcal{R}}}  
  \def\cW{{\mathcal{W}}}

\def\ba{{\mathbf{a}}} \def\bb{{\mathbf{b}}}   \def\be{{\mathbf{e}}}
\def\bff{{\mathbf{f}}}    
   \def\bn{{\mathbf{n}}} 
    
 \def\bv{{\mathbf{v}}} \def\bw{{\mathbf{w}}}

  \def\bH{{\mathbf{H}}} \def\bI{{\mathbf{I}}} 
    
   \def\bS{{\mathbf{S}}}

\def\argmin{\mathop{\mathrm{argmin}}}
\def\argmax{\mathop{\mathrm{argmax}}}

     \def\d4{\!\!\!\!}

  \def\R{{\mathbb{R}}} \def\C{{\mathbb{C}}}   \def\B{{\mathbb{B}}}

  \def\-{\! - \!}  \def\+{\! + \!}  \def\={\! = \!}  \def\>{\! > \!}

\newtheorem{lemma}{Lemma}

\newtheorem{remark}{Remark}

\newcommand{\bef}{\begin{figure}}
\newcommand{\eef}{\end{figure}}
\newcommand{\beq}{\begin{eqnarray}}
\newcommand{\eeq}{\end{eqnarray}}

\newenvironment{proof}[1][Proof]{\begin{trivlist}
\item[\hskip \labelsep {\bfseries #1}]}{\end{trivlist}}

\newcommand{\qed}{\nobreak \ifvmode \relax \else
\ifdim\lastskip<1.5em \hskip-\lastskip \hskip1.5em plus0em
minus0.5em \fi \nobreak \vrule height0.5em width0.5em
depth0.25em\fi}

\begin{document}
\title{Enhanced Beam Alignment for Millimeter Wave MIMO Systems: A Kolmogorov Model}
\author{Qiyou Duan$^1$, Taejoon Kim$^2$, and Hadi Ghauch$^3$ \\ 
\IEEEauthorblockA{ \small $^1$Department of Electrical Engineering, City University of Hong Kong, Kowloon, Hong Kong \\
                   \small $^2$Department of Electrical Engineering and Computer Science, The University of Kansas, Lawrence, KS $66045\!-\!7608$, USA\\
                   \small $^3$Department of COMELEC, Telecom-ParisTech, Paris, France\\
                   \small Email: qyduan.ee@my.cityu.edu.hk$^1$, taejoonkim@ku.edu$^2$, hadi.ghauch@telecom-paristech.fr$^3$ }}

\maketitle

\begin{abstract}
We present an enhancement to the problem of beam alignment in millimeter wave (mmWave) multiple-input multiple-output (MIMO) systems, based on a modification of the machine learning-based criterion, called Kolmogorov model (KM), previously applied to the beam alignment problem. Unlike the previous KM, whose computational complexity is not scalable with the size of the problem, a new approach, centered on discrete monotonic optimization (DMO), is proposed, leading to significantly reduced complexity. We also present a Kolmogorov-Smirnov (KS) criterion for the advanced hypothesis testing, which does not require any subjective threshold setting compared to the frequency estimation (FE) method developed for the conventional KM. Simulation results that demonstrate the efficacy of the proposed KM learning for mmWave beam alignment are presented.
\end{abstract}


\section{Introduction}  \label{Section I}
A fundamental bottleneck in operating large-dimensional millimeter wave (mmWave) array antenna systems is how to accurately align beams between the transmitter and receiver in low latency \cite{Heath16, Hur13}. The use of directional narrow beams for searching the entire beam space (also called exhaustive beam search) is an extremely time-consuming operation; the exhaustive beam search has been used in existing mmWave WiFi standards including IEEE 802.15.3c \cite{Std1} and IEEE 802.11ad \cite{Std2}, for example. For reduced overhead beam alignment, hierarchical codebooks \cite{Hur13,Alkhateeb14}, compressed sensing-based algorithms \cite{Sun17,Zhang18}, overlapped beam pattern \cite{Kokshoorn17} and beam coding \cite{Shabara19} have been proposed over the years, establishing a ``structured beam alignment'' paradigm. Despite a plethora of such beam alignment methods, the overhead issue still remains a critical challenge in mmWave communications.

Recently, the beam alignment problem has been approached in a statistical-machine-learning point-of-view \cite{Chan19}, with a primary focus on an application of the Kolmogorov model (KM) \cite{Ghauch18}. In \cite{Chan19}, Kolmogorov elementary representations (KERs) of the received signal power values that are associated with the beam pairs in a training beam codebook are learned by solving a constrained error minimization problem. In doing so, the KERs of unsounded beam pairs are predicted by exploiting the predictive power of the KM, leading to a significantly reduced beam alignment overhead. However, there are two fundamental limitations to the conventional KM learning in the beam alignment context. First, the computational complexity of the KM training algorithm in \cite{Chan19,Ghauch18} is prohibitively high; the complexity is not scalable with the number of antennas and the size of codebooks. Second, the initial work in \cite{Chan19} centers on a frequency estimation (FE) method to estimate empirical probabilities of the training set, which has to rely on a threshold setting for hypothesis testing; the threshold value is treated as a hyper-parameter, which is determined based on numerical simulations. Ultimately, the desired threshold setting must account for a specific performance criterion so as to improve the predictive power of KM.

In fact, in mmWave-based systems, quality of service is primarily dominated by latency \cite{Yang18}. In particular, the requirements of low latency and overhead are perhaps even more critical than those for high throughput. Motivated by this, we propose an enhancement to the problem of mmWave multiple-input multiple-output (MIMO) beam alignment by leveraging discrete monotonic optimization (DMO) frameworks \cite{Tuy06,Kim15}, leading to a significantly reduced amount of computational complexity compare to the previous KM \cite{Chan19}. We also propose a new threshold approach to obtaining empirical probabilities of the training set, which improves the performance of hypothesis testing for the FE of KM. Our approach is based on utilizing the Kolmogorov-Smirnov (KS) test criterion \cite{Zhang10,Marcum15}, which is desired because it can set a detection threshold without access to a priori knowledge.

The remainder of the paper is organized as follows. In Section \ref{Section II}, we introduce the system model and briefly review the related work on the KM-based beam alignment. In Section \ref{Section III}, we propose the DMO algorithm to solve the KM learning optimization problem and provide a new method building the empirical training statistics via the KS test. In Section \ref{Section IV}, simulation results are presented to illustrate the superior performance of the proposed algorithm. Finally, we conclude the paper in Section \ref{Section V}.

\section{System Model and Previous Work} \label{Section II}

We present the beam alignment system model and provide an overview of the previous work under consideration.

\subsection{System Model} \label{Subsection II-A}

Suppose a point-to-point mmWave MIMO system where an independent block fading channel with a coherence block length $T_B$ (channel uses) is assumed. The transmitter and receiver are equipped with $N_t$ and $N_r$ antennas, respectively. For simplicity, we adopt a low-complexity architecture where only one radio-frequency (RF) chain is employed at both the transmitter and receiver sides.

During a coherence block $T_B$, the transmitter and receiver intend to spend $K$ ($K\ll T_B$) channel uses to align the best transmit and receive beam pair for data transmission. To be specific, the transmitter and receiver choose an analog beamformer $\bff_t \in \C^{N_t\times 1}$ and combiner $\bw_r \in \C^{N_r\times 1}$ from the pre-designed beam sounding codebooks $\cF$ and $\cW$ such that $\bff_t \in\cF$ and $\bw_r\in\cW$, respectively. We denote the index sets of $\cF$ and $\cW$ as $\cI_{\cF}$ and $\cI_{\cW}$, respectively, with cardinalities $|\cI_{\cF}|$ and $|\cI_{\cW}|$. Assume that $\bff_t$ and $\bw_r$ are unit-norm, i.e., $\|\bff_t\|_2=\|\bw_r\|_2=1$. The received signal associated with the beam pair $(\bff_t,\bw_r)$ is therefore given by
\beq \label{received signal for (t,r)}
y_{t,r} \!=\! \bw_r^*(\bH\bff_t s_t \!+\! \bn) \!=\! \bw_r^*\bH\bff_t s_t \!+\! n_r, \forall (t,r)\! \in\! \cI_{\cF}\!\times\! \cI_{\cW},
\eeq
where $\bH\in\C^{N_r\times N_t}$ is the channel matrix and $s_t\in\C$ is the training symbol satisfying $\|\bff_t s_t\|_2^2 = 1$. $\bn\in\C^{N_r\times 1}$ is the additive complex white Gaussian noise vector with each entry independently and identically distributed (i.i.d.) as zero mean and $\sigma_n^2$ variance according to $\cC\cN(0,\sigma_n^2)$. $n_r\triangleq \bw_r^*\bn \sim \cC\cN(0,\sigma_n^2)$ is the effective additive noise, and thus, the signal-to-noise ratio (SNR) is $1/\sigma_n^2$.

Exhaustive beam alignment (beam sounding) is a widely used method: the transmitter and receiver jointly sound all the beams in $\cF$ and $\cW$ to find the optimal beam pair that maximizes the received signal power
\beq \label{optimal beam pair selection}
(\bff_{t^\star}, \bw_{r^\star}) = \argmax_{(\bff_t, \bw_r), (t,r)\in \cI_{\cF}\times \cI_{\cW}} \{\eta_{t,r}\triangleq |y_{t,r}|^2\}. \nonumber
\eeq
In fact, the training overhead for the exhaustive method is $\vert \cI_{\cF}\times \cI_{\cW} \vert$. Since the size of the codebooks $\vert \cI_{\cF}\vert$ and $\vert \cI_{\cW}\vert$ is large in mmWave cellular networks, the drastic training overhead of exhaustive beam alignment overwhelms the available coherent channel resources. To tackle this issue, a learning-based approach, KM, was proposed to reduce the beam alignment overhead while maintaining appreciable beam alignment performance \cite{Chan19}.

\subsection{Previous Work: KM-Based Beam Alignment} \label{Subsection II-B}
A binary random variable $X_{t,r}\in\{0,1\}$ is introduced to indicate the ``good'' and ``poor'' quality of the beam pair $(\bff_t,\bw_r)$ for $(t,r)\in \cI_{\cF}\times \cI_{\cW}$ as
\beq \label{binary model}
\begin{cases}
  \Pr(\eta_{t,r}\geq \tau) = \Pr(X_{t,r} = 1) \\
  \Pr(\eta_{t,r}< \tau) = \Pr(X_{t,r} = 0)
\end{cases}\!\!\!\!\!\!, \nonumber
\eeq
where $\Pr(\cE)\in [0,1]$ denotes the probability of the event $\cE$, $\tau$ is a pre-designed threshold value for the received signal power. We say that the beam pair $(\bff_t,\bw_r)$ has a ``good'' SNR, if $\eta_{t,r}\geq \tau$. Because $\Pr(X_{t,r} = 1)+\Pr(X_{t,r} = 0)=1$, it suffices to focus on the case when $X_{t,r} = 1$. The $D$-dimensional KER of $X_{t,r}$ is then defined by \cite{Ghauch18}
\beq \label{KER}
\Pr(X_{t,r}=1) = \boldsymbol{\theta}_t^T\boldsymbol{\psi}_r,\ \forall (t,r)\in \cI_{\cF}\times \cI_{\cW},
\eeq
where the probability mass function vector $\boldsymbol{\theta}_t$ is on the unit probability simplex $\cP$, i.e., $\boldsymbol{\theta}_t\in\R_+^D$ and $\boldsymbol{1}^T\boldsymbol{\theta}_t=1$, $\boldsymbol{1}$ is the all-one vector with dimension $D$, and $\boldsymbol{\psi}_r\in\B^D$ denotes the binary indicator vector of dimension $D$ such that the $d$th entry of $\boldsymbol{\psi}_r$ is $\psi_{r,d}\in\{0,1\}$.

The beam alignment using KM relies on the subsampled codebooks with index sets $\cI_{\cF}^{\text{train}}$ and $\cI_{\cW}^{\text{train}}$, such that $\cI_{\cF}^{\text{train}}\subset \cI_{\cF}$ and $\cI_{\cW}^{\text{train}}\subset \cI_{\cW}$, and have much smaller sizes, $|\cI_{\cF}^{\text{train}}|\ll |\cI_{\cF}|$ and $|\cI_{\cW}^{\text{train}}|\ll |\cI_{\cW}|$ \cite{Chan19}. We let the empirical probability that beam pair $(\bff_t,\bw_r)$ has a ``good'' SNR be $p_{t,r}$. In \cite{Chan19}, a FE method was proposed to build the training set of empirical probabilities of beam pairs in the subsampled codebooks for the KM learning algorithm, i.e., $\{p_{t,r}\}$, $\forall(t,r)\in\cI_{\cF}^{\text{train}}\times\cI_{\cW}^{\text{train}}$. Given the FE interval $T_{\text{FE}}$, the estimate of $p_{t,r}$ at time-slot $\varphi$, i.e., $p^{(\varphi)}_{t,r}$, is provided by
\beq \label{FE}
p_{t,r}^{(\varphi)} = \frac{1}{\varphi}\sum_{l=1}^{\varphi} \mathbb{I}(\eta_{t,r}^{(l)}\geq \tau),\ \varphi\in\{1,\ldots,T_{\text{FE}}\},
\eeq
where $\eta_{t,r}^{(l)}$ is the received signal power obtained by sounding the beam pair $(\bff_t, \bw_r)$ at time-slot $l\in\{1,\ldots,\varphi\}$ and $\mathbb{I}(\cdot)$ denotes the indicator function. The best FE estimate comes from $p_{t,r}^{(T_{\text{FE}})}$, which is carried out at the end of the estimation interval.

\begin{figure}[t]
\centering
\includegraphics[width=8.0cm, height=5.5cm]{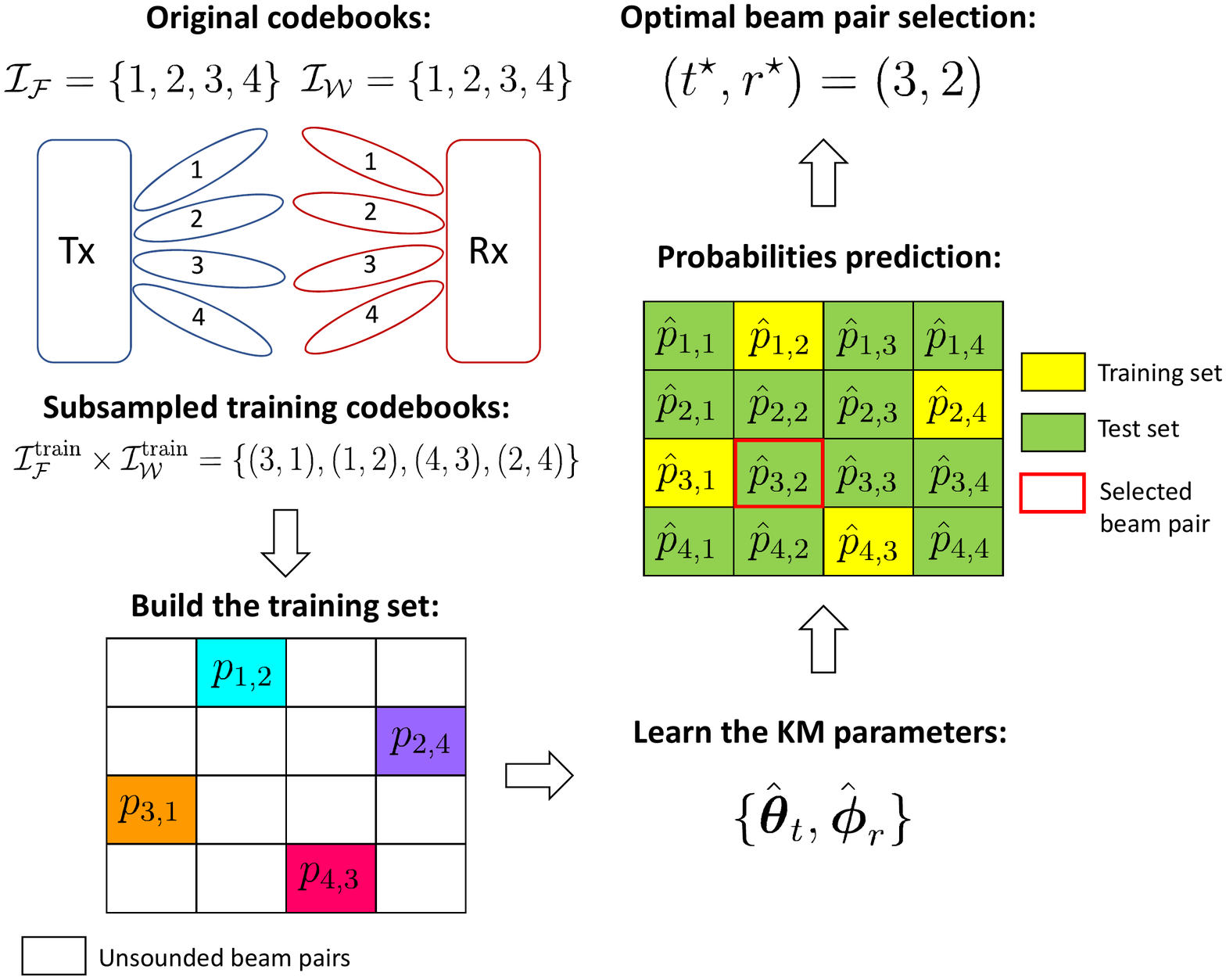}
\caption{Diagram of KM-based beam alignment $(|\cI_{\cF}|=|\cI_{\cW}|=4)$.} \label{fig.0}
\end{figure}

Once the training set (of empirical probabilities) is constructed, the KM learning algorithm proceeds to optimize the KM parameter vectors $\{\boldsymbol{\theta}_t\}$ and $\{\boldsymbol{\psi}_r\}$ by solving the constrained error minimization problem:
\beq \label{KM problem formulation}
\begin{split}
\{\hat{\boldsymbol{\theta}}_t\}, \{\hat{\boldsymbol{\psi}}_r\} \!&=\!\!\!\argmin_{\{\boldsymbol{\theta}_t\},\{\boldsymbol{\psi}_r\}}\!\!\! \sum_{(t,r)\in\cI_{\cF}^{\text{train}}\times \cI_{\cW}^{\text{train}}}\!\!\! ( \boldsymbol{\theta}_t^T\boldsymbol{\psi}_r - p_{t,r} )^2 \\
   & \text{s.t.}\ \boldsymbol{\theta}_t \in \cP, \forall t \in \cI_{\cF}^{\text{train}}, \boldsymbol{\psi}_r\in \B^D, \forall r \in \cI_{\cW}^{\text{train}}
\end{split}.
\eeq
In order to handle the coupled non-convex combinatorial optimization in \eqref{KM problem formulation}, a block-coordinate descent (BCD) method \cite{Ghauch18,Chan19} was proposed by dividing the problem in \eqref{KM problem formulation} into two subproblems:
(i) linearly-constrained  quadratic program (LCQP):
\beq \label{subproblem 1: LCQP}
\min_{\boldsymbol{\theta}_t \in \cP}\ \boldsymbol{\theta}_t^T\bS_t\boldsymbol{\theta}_t - 2\boldsymbol{\theta}_t^T\bv_t + \rho_t,
\eeq
where $\bS_t\triangleq \sum_{r \in \cI_{\cW}^{\text{train}}} \boldsymbol{\psi}_r\boldsymbol{\psi}_r^T$, $\bv_t \triangleq \sum_{r \in \cI_{\cW}^{\text{train}}} \boldsymbol{\psi}_r p_{t,r}$, and $\rho_t \triangleq \sum_{r \in \cI_{\cW}^{\text{train}}} p_{t,r}^2$, and
(ii) binary quadratic program (BQP):
\beq \label{subproblem 2: BQP}
\min_{\boldsymbol{\psi}_r\in \B^D}\ \boldsymbol{\psi}_r^T\bS_r\boldsymbol{\psi}_r - 2\bv_r^T\boldsymbol{\psi}_r + \rho_r,
\eeq
where $\bS_r\triangleq \sum_{t\in \cI_{\cF}^{\text{train}}} \boldsymbol{\theta}_t\boldsymbol{\theta}_t^T$, $\bv_r\triangleq \sum_{t\in \cI_{\cF}^{\text{train}}} \boldsymbol{\theta}_t p_{t,r}$ and $\rho_r \triangleq \sum_{t\in \cI_{\cF}^{\text{train}}} p_{t,r}^2$.
The KM solves the two subproblems in \eqref{subproblem 1: LCQP} and \eqref{subproblem 2: BQP} in an alternative way and iteratively refines the KM parameters $\{\boldsymbol{\theta}_t\}$ and $\{\boldsymbol{\psi}_r\}$.
More specifically, by exploiting the fact that the optimization in \eqref{subproblem 1: LCQP} is carried out over the unit probability simplex, a simple iterative Frank-Wolfe (FW) algorithm \cite{Jaggi13} was proposed to optimally solve \eqref{subproblem 1: LCQP}, while the semi-definite relaxation with randomization (SDRwR) was employed to optimally solve \eqref{subproblem 2: BQP} asymptotically in $D$ \cite{Kisialiou10}.

We let $\{\hat{\boldsymbol{\theta}}_t, \hat{\boldsymbol{\psi}}_r\}$ be the learned KM parameters to the problem in \eqref{KM problem formulation}. The predictive power of KM is exploited to infer the probabilities of the test set (i.e., beam pairs which are not sounded) as
\beq \label{probabilities prediction}
\hat{p}_{t,r}\triangleq \hat{\boldsymbol{\theta}}_t^T\hat{\boldsymbol{\psi}}_r,\ \forall (t,r)\in (\cI_{\cF}\times\cI_{\cW})\backslash (\cI_{\cF}^{\text{train}}\times\cI_{\cW}^{\text{train}}).
\eeq
Finally, the optimal beam pair with the highest probability of having a ``good'' SNR is selected by evaluating both the training and test sets as
\beq \label{optimal beam index pair selection}
(t^\star, r^\star)=\argmax_{(t,r)\in\cI_{\cF}\times\cI_{\cW}} \{\hat{p}_{t,r}=\hat{\boldsymbol{\theta}}_t^T\hat{\boldsymbol{\psi}}_r\}.
\eeq
A diagram of the KM-based beam alignment, which conceptually visualizes the system model and the framework, can be found in Fig. \ref{fig.0}.

\subsubsection{Desired Attributes of KM}
There are three main advantages of KM that make it superior to other data representations such as matrix factorization (MF) \cite{Koren09}, SVD-based representations \cite{Koren08}, and nonnegative MF \cite{Lee01}: (i) the fact that the KM in \eqref{KER} represents an actual probability is exploited to model the quality of beam pairs in terms of SNR, (ii) KM offers improved prediction performance over nonnegative MF \cite{Stark16}, and (iii) the interpretability of the KM in \eqref{KER} , namely, the insight that it exhibits about the data, which is not possible with other learning methods that fall under the black-box type.

\subsubsection{Main Contribution of This Work}
While the SDRwR method in solving \eqref{subproblem 2: BQP} is asymptotically optimal \cite{Chan19,Ghauch18} it demands huge computational cost and thus violates the low-latency requirement in the mmWave communications \cite{Yang18}. Moreover, the lack of an appropriate threshold design criterion of the FE method in \cite{Chan19} limits the beam alignment performance of the KM-based approach. To address the above limitations, we first propose an enhanced KM learning algorithm for beam alignment by leveraging DMO. A novel empirical probability estimation method based on the KS test is then provided with a proper threshold selection criterion. The proposed algorithm exhibits better beam alignment performance with a significantly reduced computational time compared to the existing work.

\section{Proposed Algorithm} \label{Section III}
To reduce the prohibitively high computational cost of SDRwR, in this section, a DMO framework is proposed. Moreover, a new method based on the KS test is presented.

\subsection{Discrete Monotonic Optimization} \label{Section III.A}
Prior to delivering the proposed algorithm, we provide a lemma showing an equivalent reformulation of the problem in \eqref{subproblem 2: BQP}.
\begin{lemma} \label{BQP to DMF}
The BQP problem in \eqref{subproblem 2: BQP} is equivalent to the maximization of a difference of two monotonically increasing functions and the binary constraints $\boldsymbol{\psi}_r\in \B^D$ in \eqref{subproblem 2: BQP} is equivalently transformed to continuous monotonic constraints:
\beq \label{reformulated DMF}
\begin{split}
& \max_{\boldsymbol{\psi}_r} \left\{f(\boldsymbol{\psi}_r) = f^+(\boldsymbol{\psi}_r) - f^-(\boldsymbol{\psi}_r) \right\} \\
& \text{s.t.}\ g(\boldsymbol{\psi}_r) - h(\boldsymbol{\psi}_r) \leq 0, \boldsymbol{\psi}_r \in [\boldsymbol{0}, \boldsymbol{1}]
\end{split}\ ,
\eeq
where $f^+(\boldsymbol{\psi}_r)\triangleq 2\bv_r^T\boldsymbol{\psi}_r$, $f^-(\boldsymbol{\psi}_r)\triangleq \boldsymbol{\psi}_r^T\bS_r\boldsymbol{\psi}_r$, $g(\boldsymbol{\psi}_r)\triangleq \sum_{d=1}^{D}\psi_{r,d}$, $h(\boldsymbol{\psi}_r)\triangleq \sum_{d=1}^{D}\psi_{r,d}^2$, and $\boldsymbol{\psi}_r \in [\boldsymbol{0}, \boldsymbol{1}]$ indicates that $0\leq \psi_{r,d} \leq 1$ for every $d=1,\ldots,D$.
\end{lemma}
\begin{proof}
  Given the definition of $f^+$ and $f^-$ in \eqref{reformulated DMF}, the objective function $f$ in \eqref{reformulated DMF} is attained by transforming the minimization to the maximization and discarding the constant $\rho_r$ in \eqref{subproblem 2: BQP}. Also, $f^+$ and $f^-$ are both increasing functions with respect to $\boldsymbol{\psi}_r\in [\boldsymbol{0}, \boldsymbol{1}]$ because $\bv_r>\boldsymbol{0}$ and $\bS_r$ is a positive semi-definite matrix. The binary constraints $\psi_{r,d}\in\{0,1\}$, $d=1,\ldots,D$, can be equivalently rewritten as $\sum_{d=1}^{D}\psi_{r,d}(1-\psi_{r,d})\leq 0$, $\psi_{r,d}\in[0,1]$, $\forall d$, i.e., $g(\boldsymbol{\psi}_r) - h(\boldsymbol{\psi}_r) \leq 0$, $\boldsymbol{\psi}_r \in [\boldsymbol{0}, \boldsymbol{1}]$ in \eqref{reformulated DMF}, where $g$ and $h$ are increasing on $\R_+^D$. This completes the proof.
\end{proof}

The BQP problem in \eqref{subproblem 2: BQP} cannot be directly handled due to the discrete constraints. In \cite{Chan19}, this nuisance has been tackled by using SDRwR, which incurs impractical computational complexity. Unlike SDRwR, the equivalent problem formulation leveraging the difference of monotonic functions (DMF) in \eqref{reformulated DMF} disinvolves the intractable discrete constraints without any relaxation. Motivated by Lemma \ref{BQP to DMF}, we propose to use a branch-reduce-and-bound (BRB) approach \cite{Tuy06} to directly solves \eqref{reformulated DMF} without any relaxation and/or randomization. As will be seen in Fig. \ref{fig.1} in Section \ref{Section IV}, the proposed DMO algorithm can substantially reduce the computational complexity (two-orders-of-magnitude improvement in time complexity).
We introduce the following three main steps at each iteration in the proposed DMO algorithm, where the overall procedure is presented in detail in Algorithm \ref{DMO-based Algorithm}.

\begin{algorithm}[t]
\caption{DMO Algorithm} \label{DMO-based Algorithm}
\begin{algorithmic}[1]
\Require
$\bS_r$, $\bv_r$, and $D$.
\Ensure
$\boldsymbol{\psi}_r^\star$.
\State Initialization: Set iteration number $i=1$. Let $\cP_i = \{M\}$, $M = [\boldsymbol{0}, \boldsymbol{1}]$, $\cR_i = \phi$, and $\nu = f(\boldsymbol{0}) = 0$.
\State Reduction: Reduce each box in $\cP_i$ according to \eqref{reduced box-lower vertex} and \eqref{reduced box-upper vertex} to obtain $\cP_i^\prime = \{[\ba^\prime, \bb^\prime]|[\ba, \bb]\in\cP_i\}$. \label{returning step}
\State Bounding: Calculate $\mu(M^\prime)$ in \eqref{bounding} for each $M^\prime\in\cM_i\triangleq\cP_i^\prime\cup\cR_i$.
\State Find the feasible solution: $\boldsymbol{\psi}_r^{(i)}=\argmax_{\boldsymbol{\psi}_r}\{f(\boldsymbol{\psi}_r)>\nu| \boldsymbol{\psi}_r=\lceil(\ba^\prime+\bb^\prime)/2\rceil, M^\prime=[\ba^\prime, \bb^\prime]\in\cM_i\}$. \label{Step: find the feasible solution}
\State Update current best value: If $\boldsymbol{\psi}_r^{(i)}$ in Step \ref{Step: find the feasible solution} exists, update $\nu$ as $\nu = f(\boldsymbol{\psi}_r^{(i)})$; otherwise, $\boldsymbol{\psi}_r^{(i)}=\boldsymbol{\psi}_r^{(i-1)}$ and $\nu$ doesn't change.
\State Discarding: Delete every $M^\prime\in\cM_i$ such that $\mu(M^\prime)<\nu$ and let $\cR_{i+1}$ be the collection of remaining boxes.
\If {$\cR_{i+1}=\phi$} terminate and
\Return $\boldsymbol{\psi}_r^\star = \boldsymbol{\psi}_r^{(i)}$.
\Else
\State Let $M^{(i)} = \argmax_{M^\prime}\{\mu(M^\prime)|M^\prime\in\cR_{i+1}\}$.
   \If {$\nu\geq \varepsilon \mu(M^{(i)})$} $\varepsilon$-accuracy is reached and \Return $\boldsymbol{\psi}_r^\star = \boldsymbol{\psi}_r^{(i)}$.
   \Else
   \State Branching: Divide $M^{(i)}$ into $M^{(i)}_1$ and $M^{(i)}_2$ according to \eqref{branched box 1} and \eqref{branched box 2}.
   \State Update $\cR_{i+1}$ and $\cP_{i+1}$: $\cR_{i+1}=\cR_{i+1}\backslash M^{(i)}$ and $\cP_{i+1}=\{M^{(i)}_1, M^{(i)}_2\}$.
   \EndIf
\EndIf
\State $i=i+1$ and \Return to Step \ref{returning step}.
\end{algorithmic}
\end{algorithm}

\subsubsection{Reduction}
We let $M=[\ba, \bb]$ be one of the boxes that contain feasible solutions to \eqref{reformulated DMF} and $\nu$ be the current maximum value of the objective function $f$ in \eqref{reformulated DMF}. The reduced box $M^\prime=[\ba^\prime, \bb^\prime]\subset[\ba, \bb]$ can be defined by new lower and upper vertices $\ba^\prime$ and $\bb^\prime$, respectively, without excluding any feasible solution $\boldsymbol{\psi}_r\in[\ba, \bb]$, while maintaining $f(\boldsymbol{\psi}_r)\geq\nu$ \cite{Tuy06} as
\beq
\ba^\prime = \bb - \sum_{d=1}^{D}\alpha_d(b_d - a_d)\be_d, \label{reduced box-lower vertex} \\
\bb^\prime = \ba^\prime + \sum_{d=1}^{D}\beta_d(b_d - a^\prime_d)\be_d, \label{reduced box-upper vertex}
\eeq
where $\alpha_d=\sup\{\alpha|\alpha\in[0,1], g(\ba)- h(\bb-\alpha(b_d-a_d)\be_d)\leq 0, f^+(\bb-\alpha(b_d-a_d)\be_d)-f^-(\ba)\geq \nu\}$ and $\beta_d = \sup\{\beta|\beta\in[0,1], g(\ba^\prime + \beta(b_d - a^\prime_d)\be_d) - h(\bb) \leq 0, f^+(\bb) - f^-(\ba^\prime + \beta(b_d - a^\prime_d)\be_d)\geq \nu\}$ for $d=1,\ldots,D$, where $\be_d$ is the $d$th column of the $D$-dimensional identity matrix $\bI_D$. Note that the optimal values of $\alpha_d$ and $\beta_d$ can be found by referring to the compactness of $\alpha, \beta\in[0,1]$ and utilizing the monotonicity of $f^+$, $f^-$, $g$, and $h$ (for instance, by using a bisection method) \cite{Kim15}.

\subsubsection{Bounding}
For every reduced box $M^\prime$, an upper bound of $\nu(M^\prime)\triangleq \max\{f(\boldsymbol{\psi}_r)|$ $g(\boldsymbol{\psi}_r)-h(\boldsymbol{\psi}_r)\leq0, \boldsymbol{\psi}_r\in M^\prime\cap [\boldsymbol{0}, \boldsymbol{1}]\}$ is calculated such that
\beq \label{bounding}
\nu(M^\prime)\leq \mu(M^\prime)= f^+(\bb^\prime) - f^-(\ba^\prime).
\eeq
The upper bound $\mu(M^\prime)$ in \eqref{bounding} holds because $f^+$ and $f^-$ are monotonically increasing functions. Furthermore, $\mu(M^\prime)$ ensure $\lim_{k\rightarrow\infty} \mu(M^\prime_k) = f(\boldsymbol{\psi}_r^\star)$, where $\{M^\prime_k\}$ stands for any infinite nested sequence of boxes and $\boldsymbol{\psi}_r^\star$ is the optimal solution to \eqref{reformulated DMF}. At each iteration, any box $M^\prime$ with $\mu(M^\prime)<\nu$ is deleted because such a box does not contain $\boldsymbol{\psi}_r^\star$ anymore.

\subsubsection{Branching}
At the end of each iteration, the box with the maximum upper bound, denoted by $M^\star=[\ba^\star, \bb^\star]$, is selected and branched to accelerate the convergence of the algorithm. The box $M^\star$ is divided into two boxes
\beq
M_1^\star = \{\boldsymbol{\psi}_r\in M^\star|\psi_{r,j}\leq \lfloor c^\star_j \rfloor\}, \label{branched box 1}\\
M_2^\star = \{\boldsymbol{\psi}_r\in M^\star|\psi_{r,j}\geq \lceil c^\star_j \rceil\}, \label{branched box 2}
\eeq
where $j=\argmax_{d=1,\ldots,D} (b^\star_d - a^\star_d)$, $c^*_j = (a^\star_j + b^\star_j)/2$, $\lfloor\cdot\rfloor$ and $\lceil\cdot\rceil$ represent the element-wise floor and ceiling operations, respectively.

\begin{algorithm}[t]
\caption{Enhanced KM Learning for Beam Alignment} \label{Overall KM-learning-based Beam-alignment Algorithm}
\begin{algorithmic}[1]
\Require
$\cF$, $\cW$, $\cI_{\cF}^{\text{train}}$, $\cI_{\cW}^{\text{train}}$, $D$, $L$, $\alpha$, and $T_{KS}$.
\Ensure
$(t^\star, r^\star)$.
\State Estimate the empirical probabilities via KS test:
\For{each $\varphi=1,\ldots,T_{\text{KS}}$}
\For{each beam-index pair $(t,r)\in\cI_{\cF}^{\text{train}}\times\cI_{\cW}^{\text{train}}$}
\State Train the beam pair $(\bff_t, \bw_r)$ and obtain $Z_{t,r}^{(l)}$, $l\in\{1,\ldots,\varphi\}$ as in \eqref{detection statistics} based on $[\eta_{t,r}^{(1)}, \cdots, \eta_{t,r}^{(L)}]$.
\State Compute the empirical probabilities according to \eqref{empirical probability estimate via K-S test}.
\EndFor
\EndFor
\State Learn the KM parameters:
\For{$i=1,\ldots,I$}
\State $1)$ Update $\boldsymbol{\theta}_t^{(i)}$ via the FW algorithm \cite{Jaggi13};
\State $2)$ Update $\boldsymbol{\psi}_r^{(i)}$ via Algorithm \ref{DMO-based Algorithm}.
\EndFor
\State Obtain the final estimate $\{\hat{\boldsymbol{\theta}}_t=\boldsymbol{\theta}_t^{(I)}, \hat{\boldsymbol{\psi}}_r=\boldsymbol{\psi}_r^{(I)}\}$.
\State Compute the predicted probability for the beam pairs which are not trained yet based on \eqref{probabilities prediction}.
\State Determine the optimal beam index pair as in \eqref{optimal beam index pair selection}.
\State \Return $(t^\star, r^\star)$.
\end{algorithmic}
\end{algorithm}

The DMF optimization problem in \eqref{reformulated DMF} is solved by iteratively executing the latter three procedures until it converges within $\varepsilon$-accuracy as shown in Algorithm \ref{DMO-based Algorithm}.

\subsection{Kolmogorov-Smirnov Test} \label{Section III.B}
The choice of $\tau$ in \eqref{FE} has a profound impact on the beam alignment performance of the KM-based approach. The threshold value $\tau$ has been chosen subjectively based on numerical simulations \cite{Chan19}, which can substantially vary depending on the channel conditions and operating SNR. There lacks an appropriate selection criterion due in part to the fact that the statistics of $\eta_{t,r}$ are unknown in practice. We overcome this difficulty by proposing, in this subsection, to estimate the trained empirical probabilities $\{p_{t,r}\}$ by applying the detection-theoretic criterion for threshold setting introduced by Kolmogorov and Smirnov \cite{Millard1967,Zhang10}.

We first define the binary hypotheses of a beam pair $(\bff_t, \bw_r)$, $\forall (t,r)\in\cI_{\cF}^{train}\times\cI_{\cW}^{train}$ according to the signal model in \eqref{received signal for (t,r)} as
\beq \label{hypotheses}
\begin{split}
   &\cH_0: \ \eta_{t,r}=|n_r|^2 \\
   &\cH_1: \ \eta_{t,r}=|\bw_r^*\bH\bff_t s_t + n_r|^2,
\end{split} \nonumber
\eeq
where the null hypothesis $\cH_0$ is declared when $\eta_{t,r}$ relies on noise only and the alternative hypothesis $\cH_1$ is true when $\eta_{t,r}$ is a function of both the signal and noise.
While, under $\cH_0$, given $n_r\sim \cC\cN(0,\sigma_n^2)$, the theoretical cumulative distribution function (CDF) of $\eta_{t,r}$ is given by
\beq \label{theoretical CDF under H0}
F(\eta_{t,r}|\cH_0) = 1 - e^{-\frac{\eta_{t,r}}{\sigma_n^2}}, \nonumber
\eeq
the test statistics under $\cH_1$ is unknown. To circumvent this difficulty, KS test forms the empirical CDF of $\eta_{t,r}$ from the observed data samples $\eta_{t,r}^{(1)}, \cdots, \eta_{t,r}^{(L)}$,
\beq \label{empirical CDF}
F_L(x) = \frac{1}{L}\sum_{\ell=1}^{L} \mathbb{I}(\eta_{t,r}^{(\ell)}\leq x), \nonumber
\eeq
where $L$ denotes the number of the data samples in the KS test, which is distinguished from the time interval $T_{\text{FE}}$ in \eqref{FE}.

The KS criterion to estimate the best sample point is given by
\beq \label{detection statistics}
Z_{t,r} = \max_{x\in\R} |F_L(x) - F(x|\cH_0)|.
\eeq
The binary hypothesis test is then $Z_{t,r} \mathop{\gtrless}\limits_{\cH_0}^{\cH_1} \epsilon$,
where $\epsilon$ is the KS threshold value. Similar to conventional Neyman-Pearson, the threshold $\epsilon$ is chosen to meet the target false alarm rate $\alpha$ such that
\beq \label{false alarm rate}
\alpha \triangleq \Pr(Z_{t,r}\geq\epsilon|\cH_0) \approx 2e^{-2L\epsilon^2}, \nonumber
\eeq
where the last step is due to the Kolmogorov approximation \cite{Papoulis02}. The approximation becomes tight as $L$ tends to large such that the KS threshold can be determined by $\epsilon=\sqrt{-\frac{\ln(\alpha/2)}{2L}}$.
Finally, similar to \eqref{FE}, the KS-estimated empirical probability at time-slot $\varphi$ for any beam index pair $(t,r)\in\cI_{\cF}^{\text{train}}\times\cI_{\cW}^{\text{train}}$ is, therefore, given by
\beq \label{empirical probability estimate via K-S test}
p_{t,r}^{(\varphi)} = \frac{1}{\varphi}\sum_{l=1}^{\varphi}\mathbb{I}(Z^{(l)}_{t,r}\geq \epsilon), \ \varphi\in\{1,\ldots,T_{\text{KS}}\},
\eeq
where $Z^{(l)}_{t,r}$ is the detection statistic obtained by \eqref{detection statistics} at time-slot $l\in\{1,\ldots,\varphi\}$ and $T_{\text{KS}}$ denotes the KS estimation interval.

\begin{remark}
The key implication of the KS criterion in \eqref{detection statistics} is three folds: (i) the maximum value $Z_{t,r}$ converges to $0$ almost surely when $L$ tends to infinity if the data samples follows the distribution $F(\eta_{t,r}|\cH_0)$, (ii) the distribution of $Z_{t,r}$ does not depend on the underlying CDF being tested, and (iii) the maximum of difference between the CDFs stands for a jump/concentration in probability and thus becomes more representative to tell the difference of distribution compared to other statistics such as minimum and median.
\end{remark}

Incorporating Algorithm \ref{DMO-based Algorithm} to solve the BQP in \eqref{subproblem 2: BQP} and the KS test in \eqref{empirical probability estimate via K-S test} to estimate the empirical probabilities in \eqref{FE} instead of FE, we are ready to elucidate the overall proposed beam alignment procedure in Algorithm \ref{Overall KM-learning-based Beam-alignment Algorithm}.

\begin{figure}[t]
\centering
\includegraphics[width=7.2cm, height=5.2cm]{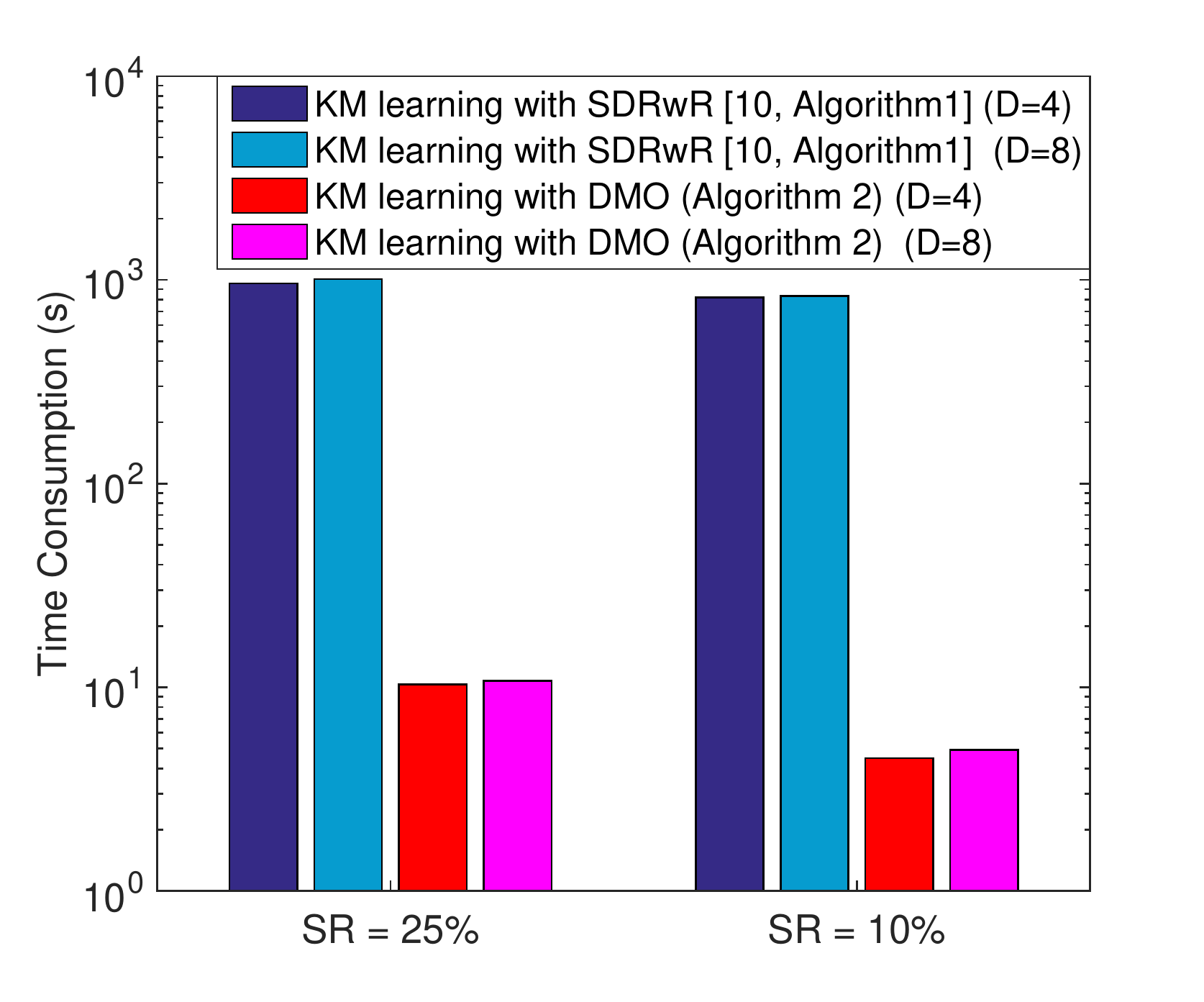}
\caption{Time consumption comparison between the conventional KM learning with SDRwR \cite[Algorithm 1]{Chan19} and the proposed KM learning with DMO (i.e., Algorithm \ref{Overall KM-learning-based Beam-alignment Algorithm}) ($N_t = N_r = |\cI_{\cF}| = |\cI_{\cW}| = 16$, $T_{\text{FE}}=T_{\text{KS}}=8$, $L=5$, $\alpha=0.05$ and $\tau=12$ dB).} \label{fig.1}
\end{figure}

\section{Simulation Results} \label{Section IV}
In this section, we provide the numerical results of the proposed beam alignment approach in mmWave MIMO channels. We adopt the physical representation of sparse mmWave MIMO channels \cite{Alkhateeb14,Heath16} and assume that the rank of the channel matrix is $1$. We set $N_t = N_r = |\cI_{\cF}| = |\cI_{\cW}|$, $T_{\text{FE}} = T_{\text{KS}} = 8$, and $L=5$ throughout the simulation. The sampling rate, defined as the ratio of the number of beam pairs in the subsampled training codebook to the total number of the beam pairs in the original codebook, is given by $\text{SR} = |\cI_{\cF}^{\text{train}}\times\cI_{\cW}^{\text{train}}|/|\cI_{\cF}\times\cI_{\cW}|$. We obtain the numerical results by conducting $100$ Monte Carlo simulations.

\begin{figure}[t]
\centering
\includegraphics[width=7.2cm, height=5.8cm]{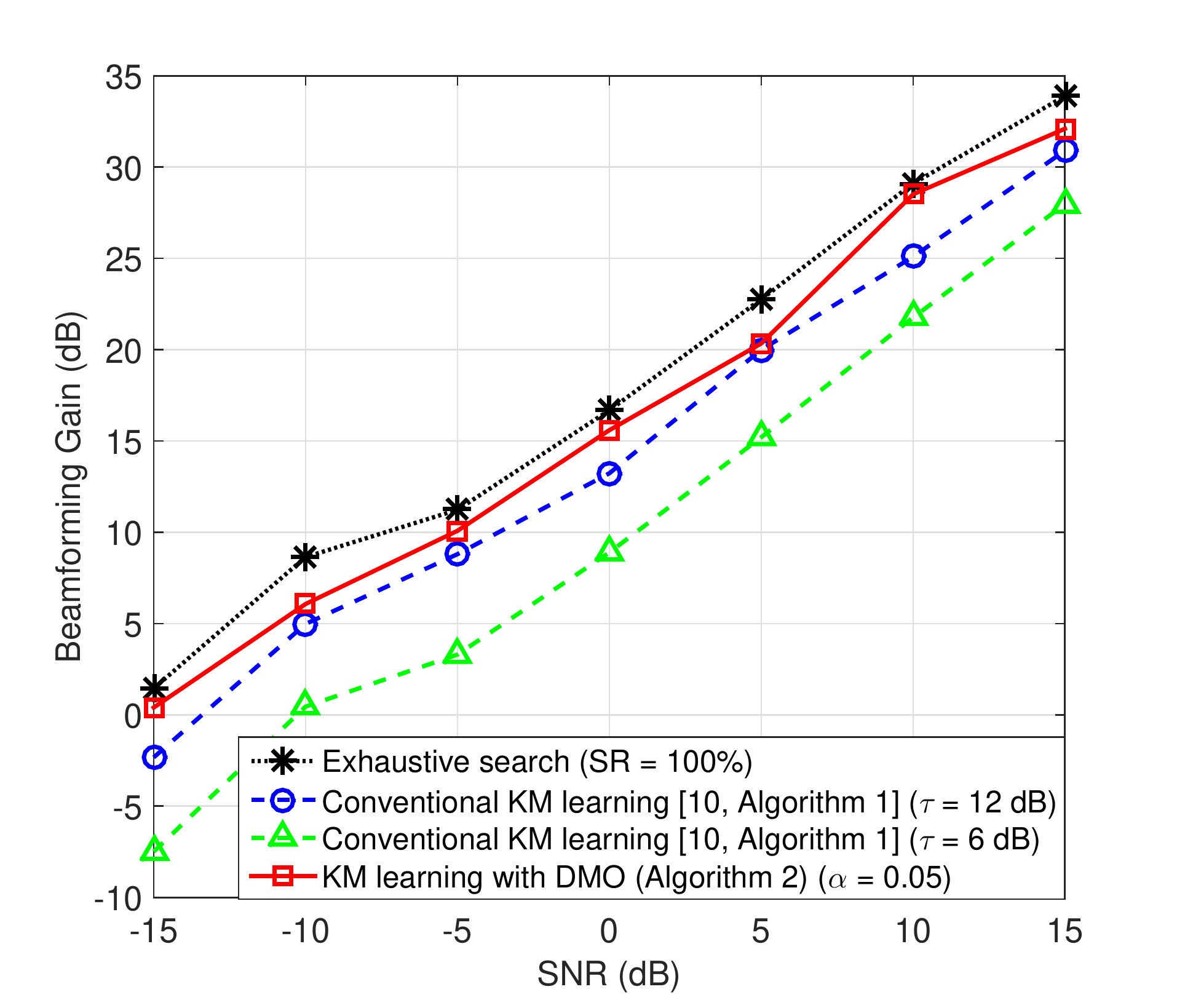}
\caption{Beamforming gain comparison between the conventional KM learning \cite[Algorithm 1]{Chan19} and Algorithm \ref{Overall KM-learning-based Beam-alignment Algorithm} ($N_t = N_r = |\cI_{\cF}| = |\cI_{\cW}| = 16$, $D=8$, $T_{\text{FE}}=T_{\text{KS}}=8$, $L=5$, and $\text{SR}=25\%$).} \label{fig.2}
\end{figure}

In Fig. \ref{fig.1}, the average time (in seconds) consumed to execute Algorithm \ref{Overall KM-learning-based Beam-alignment Algorithm} (i.e., the proposed KM learning with DMO) is compared with the conventional KM learning with SDRwR (i.e., Algorithm $1$ in \cite{Chan19}) for $\text{SR}=25\%,~10\%$ and $D=4,~8$, respectively. Notice that we measure the running time by using ``cputime'' function in MATLAB. We set the target false alarm rate $\alpha=0.05$ for the KS test in Algorithm \ref{Overall KM-learning-based Beam-alignment Algorithm} and $\tau=12$ dB for the FE in \cite[Algorithm 1]{Chan19} to obtain the empirical probabilities for the training set. We further assume $N_t = N_r = |\cI_{\cF}| = |\cI_{\cW}| = 16$ here.
It is clear from Fig. \ref{fig.1} that the proposed Algorithm \ref{Overall KM-learning-based Beam-alignment Algorithm} substantially accelerates the computational speed compared to the conventional KM learning with SDRwR \cite[Algorithm 1]{Chan19}; more than $100$ times of improvement is observed.

\begin{figure}[t]
\centering
\includegraphics[width=7.2cm, height=5.8cm]{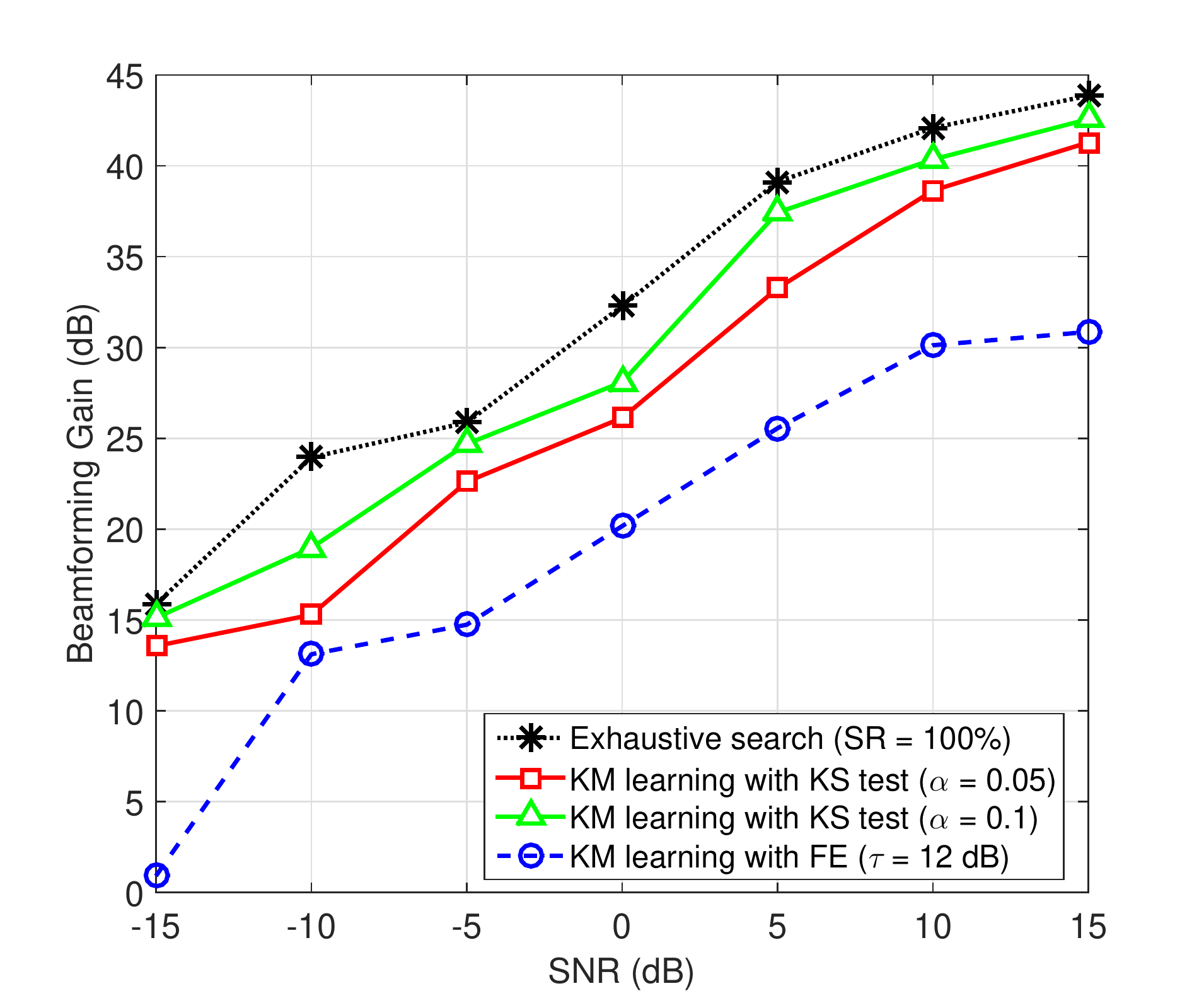}
\caption{Beamforming gain comparison between Algorithm \ref{Overall KM-learning-based Beam-alignment Algorithm} and the KM learning with FE ($N_t = N_r = |\cI_{\cF}| = |\cI_{\cW}| = 64$, $D=8$, $T_{\text{FE}}=T_{\text{KS}}=8$, $L=5$, and $\text{SR}=25\%$).} \label{fig.3}
\end{figure}

In Fig. \ref{fig.2}, the average beamforming gains of the conventional KM learning algorithm \cite[Algorithm 1]{Chan19} and the proposed Algorithm \ref{Overall KM-learning-based Beam-alignment Algorithm} are evaluated for $N_t = N_r = |\cI_{\cF}| = |\cI_{\cW}| = 16$, $D=8$, and $\text{SR}=25\%$, where given the selected beam pair $(\bff_{t^\star}, \bw_{r^\star})$, based on each algorithm, the beamforming gain is calculated from \eqref{received signal for (t,r)} by $G_{t^\star,r^\star}=\|\bw^*_{r^\star}\bH\bff_{t^\star}\|_2^2/\sigma_n^2$. In Fig. \ref{fig.2}, the curves of the conventional KM learning are evaluated for different threshold values $\tau=6,~12$ dB, while the curve of Algorithm \ref{Overall KM-learning-based Beam-alignment Algorithm} is evaluated for $\alpha=0.05$. Moreover, the performance of the exhaustive search, a benchmark, consuming $|\cI_{\cF}\times\cI_{\cW}|$ channel uses for the beam alignment, is also presented. As can be seen from Fig. \ref{fig.2},  Algorithm \ref{Overall KM-learning-based Beam-alignment Algorithm} shows an improvement compared to the conventional KM learning with substantially reduced complexity.

The efficacy of the proposed KS test in improving the proposed KM learning capability is further evaluated. In Fig. \ref{fig.3}, we show the beamforming gain of Algorithm \ref{Overall KM-learning-based Beam-alignment Algorithm} and the one by replacing the KS test in Algorithm \ref{Overall KM-learning-based Beam-alignment Algorithm} with the FE as shown in \eqref{FE} for $N_t = N_r = |\cI_{\cF}| = |\cI_{\cW}| = 64$, $D=8$, and $\text{SR}=25\%$. Fig. \ref{fig.3} illustrates that, with a false alarm rate guarantee, the proposed KS test substantially improves the learning capability of the KM.

\section{Conclusions} \label{Section V}
In this paper, we proposed an enhanced KM learning algorithm for beam alignment in mmWave MIMO channels. Based on DMO, one key step in learning the KM parameters, i.e., the BQP, was substantially accelerated. By considering the uncertainty brought by FE due to subjective threshold setting, the KS test was proposed to obtain the empirical probabilities of the training set, based on the detection-theoretic criterion. The simulation results demonstrate that the proposed KM learning with DMO and KS shows better beam alignment performance with a substantially reduced computational complexity compared to the conventional KM algorithm.

\vskip 1\baselineskip
\vskip 1\baselineskip

\bibliographystyle{IEEEtran}
\bibliography{IEEEabrv,draft_KM}

\end{document}